\documentclass{snapshotmfo}
\usepackage[utf8]{inputenc}
\usepackage[dvipsnames]{xcolor}
\usepackage{mathtools}
\usepackage{enumitem}
\usepackage{amsmath,amssymb}
\usepackage{amsthm}
\usepackage{tcolorbox}
\usepackage{epigraph}
\usepackage{hyperref}
\hypersetup{
    colorlinks=true,
    linkcolor=Aquamarine,
    filecolor=Aquamarine,      
    urlcolor=Aquamarine,
    citecolor=Aquamarine,
}

\usepackage[round]{natbib}
\usepackage[USenglish]{babel}

\usepackage{ellipsis}
\providecommand{\keywords}[1]{\textbf{\textit{Keywords:}} #1}
\providecommand{\jel}[1]{\textbf{\textit{JEL Classifications:}} #1}
\DeclareMathOperator\supp{supp}
\DeclareMathOperator\inter{int}

\DeclareMathOperator*{\conv}{conv}
\DeclareMathOperator*{\extreme}{extreme}

\newtheorem{theorem}{Theorem}[section]
\newtheorem{proposition}[theorem]{Proposition}
\newtheorem{lemma}[theorem]{Lemma}

\newtheorem{corollary}[theorem]{Corollary}

\newtheorem{definition}[theorem]{Definition}

\author{Mark Whitmeyer\thanks{Arizona State University. Email: \href{mailto:mark.whitmeyer@gmail.com}{mark.whitmeyer@gmail.com}. I thank Wenhao Li, Joseph Whitmeyer, Tom Wiseman, and Kun Zhang for their comments.}}
\title{Can One Hear the Shape of a Decision Problem?}

\begin{document}
\begin{abstract}
We explore the connection between an agent's decision problem and her ranking of information structures. We find that a finite amount of ordinal data on the agent's ranking of experiments is enough to identify her (finite) set of undominated actions (up to relabeling and duplication) and  the beliefs rendering each such action optimal. An additional smattering of cardinal data, comparing the relative value to the agent of finitely many pairs of experiments, identifies her utility function up to an action-independent payoff.
\end{abstract}
\\
\jel{C72; D81; D82; D83}

\section{Motivation}

To what extent does an agent's ranking of information structures specify the decision problem she faces? We know she is a Bayesian expected-utility maximizer and that her conception of the world has finitely many states. We, furthermore, know her prior (though this is inessential) and that only finitely many of her available actions are ever optimal; \textit{viz.,} are undominated. What we do not know is her utility function itself.

The goal of this work is to understand the connection between an agent's value for information and her decision problem. Obviously, knowing an agent's decision problem and prior fully specifies her value for information: her (expected) utility gain from observing any Blackwell experiment (\cite*{blackwell}, \cite*{blackwell2}) can be computed with ease. Any two experiments can, therefore, be ranked. Here, we study the inverse problem: suppose we know a finite collection of the agent's comparative ranks of experiments. Does that pin down the agent's utility function?

The three main results of this paper are two positive ones, with a negative one sandwiched in between. In the first, Lemma \ref{cidentify}, we discover that a finite collection of ranked experiments identifies the agent's undominated actions (up to relabeling and duplication) as well as the regions of beliefs on which each such action is optimal. The second result, Proposition \ref{nonidentification}, reveals an important indeterminacy about the agent's utility function that cannot be resolved with only ordinal rankings of experiments. The third result, Theorem \ref{theorem}, states that an additional finite collection of utility differences between experiments (the difference in utils from observing one experiment versus another) identifies the agent's utility function up to a decision-irrelevant payoff.

These results show that an agent's preferences for information reveal a significant amount about her decision problem. A finite amount of ordinal information alone is enough to completely characterize the agent's behavior from observing any information structure. An additional finite number of utility differences yields the agent's exact value for information, for any comparison of information structures. The rationale behind these results is that expected utility maximization imposes a significant amount of structure on how an agent values information.

As two-player (simultaneous-move) games, fixing the behavior of one player, are decision problems from the perspective of the other player, we also find that a finite collection of statements about the intersections of best-response sets is enough to characterize a two-player game, up to inessential details. That is, such a finite collection identifies each player's set of rationalizable actions (up to duplication and relabeling), as well as the sets of beliefs rationalizing each such action. This, in turn, pins down the Nash equilibria of the game.

\section{The Question(s) at Hand}

There is a decision-maker (DM) encountering a static decision problem, which we define to be a triple \(\left(\Theta, A, u\right)\) consisting of a finite set of possible states \(\Theta\), a nonempty set of actions \(A\), and a bounded utility function \(u \colon A \times \Theta \to \mathbb{R}\). We, the observers, know \(\Theta\) and the DM's full support prior \(\mu \in \Delta\left(\Theta\right) \equiv \Delta\). We also know that the DM is a Bayesian expected-utility maximizer whose utility function is bounded, only finitely many of her actions are undominated, and no two undominated actions have identical payoffs in every state.\footnote{Alternatively, we resign ourselves to being unable to distinguish between actions that are exact duplicates. This is a genericity requirement.}

Furthermore, we also have a finite collection of the DM's rankings of signals. That is, there is a finite collection of Blackwell experiments \(\pi_1, \dots, \pi_m\) and rankings of the form \(\pi_i \succ \pi_j\) and \(\pi_k \sim \pi_l\) (\(i, j, k, l \in \left\{1, \dots, m\right\}\)); meaning the DM strictly prefers experiment \(\pi_i\) to \(\pi_j\) and is indifferent between experiments \(\pi_k\) and \(\pi_l\), respectively.

\begin{tcolorbox}[colframe=Aquamarine,colback=white] \textbf{Identifying the decision problem:} is this data--the finite collection of ranked experiments, the prior, that only finitely many actions are undominated, etc.--enough to pin down the decision problem up to an action-irrelevant payoff (and to relabeling of the actions)?
\end{tcolorbox}

If this data is insufficient, will a finite supplement of cardinal information suffice in the following sense? For a Blackwell experiment \(\pi\) and a prior \(\mu\), let \(X\) be the random posterior produced by Bayesian updating, which is distributed according to \(F = B(\mu,\pi)\) (where \(B\) is the Bayes map).\footnote{We borrow this notation from \cite*{denti1} and \cite*{denti2}.} Then, let \(W(\pi) \in \mathbb{R}\) denote the DM's Expected payoff in her decision problem from first observing the outcome of experiment \(\pi\) then taking an optimal action: 
\[W(\pi) = \mathbb{E}_{F}\max_{a \in A} \mathbb{E}_x u\left(a,\theta\right)\text{.}\]
Suppose we also have a finite collection of Blackwell experiments \(\pi_1, \dots, \pi_m\), real numbers \(\zeta_1, \dots, \zeta_k\), and equalities \(W(\pi_i) = W(\pi_j) + \zeta_l\). Note that we do not know \(W\), this data consists of statements ``the DM's utility gain from \(\pi_j\) to \(\pi_i\) is \(\zeta_l\) utils.'' Now can we identify the state-dependent utility up to a decision-irrelevant random payoff (and to relabeling of the actions)?

\subsection{Recasting the Question}

Given the set of states \(\Theta\), a set of actions \(A\), and a bounded state-dependent utility function \(u \colon A \times \Theta \to \mathbb{R}\), the DM's value function in belief \(x \in \Delta\) is
\[V\left(x\right) \coloneqq \max_{a \in A}
\mathbb{E}_{x}u\left(a,\theta\right) \text{.}\]
There are \(n \geq 2\) states (\(\left|\Theta\right| = n\)), so \(\Delta \subset \mathbb{R}^{n-1}\). 

An action \(a \in A\) is undominated if there exists a belief \(x \in \Delta\) such that  \[\mathbb{E}_{x}u\left(a,\theta\right) > \mathbb{E}_{x}u\left(b,\theta\right), \ \text{for all} \ b \in A \setminus \left\{a\right\}\text{.}\]
As only finitely many actions are undominated, \(V\) is the maximum of finitely many hyperplanes, i.e., is convex and piecewise affine.

As noted before, via Bayes' law, given prior \(\mu \in \Delta\) any experiment \(\pi\) corresponds to a random vector \(X\) supported on \(\Delta\) whose expectation is the DM's prior \(\mu\). Accordingly, we have a finite collection (\(m \in \mathbb{Z}_{++}\)) of equalities and strict inequalities of the form \(\mathbb{E}_{F_i} V(X_i) > \mathbb{E}_{F_j} V(X_j)\) and \(\mathbb{E}_{F_k} V(X_k) = \mathbb{E}_{F_l} V(X_l)\) (\(i, j, k, l \in \left\{1, \dots, m\right\}\)). We call this a collection of ordered expectations.

\begin{tcolorbox}[colframe=Aquamarine,colback=white] \textbf{Identifying the decision problem:} is this data--the collection of ordered expectations, and that \(V\) is the maximum of finitely many affine functions--enough to pin down the decision problem? That is, is there such a collection that is satisfied by a piecewise-affine \(V\); and, moreover, if \(W\) satisfies the collection, then \(W = V + \varphi\), where \(\varphi \colon \Delta \to \mathbb{R}\) is affine? What if we also have a finite supplement of additional \textit{cardinal} data of the form \(\mathbb{E}_{F_i} V(X_i) = \mathbb{E}_{F_j} V(X_j) + \zeta_l\)? Will that pin down \(V + \varphi\)?
\end{tcolorbox}

\section{Answering the Question}

Any such value function \(V\) can be projected onto the belief simplex, \(\Delta\), which yields a finite collection \(C\) of polytopes of full dimension \(C_i\) (\(i = 1, \dots, t\)), where \(t\) is the number of undominated actions in \(A\). Formally, for each undominated \(a_i \in A\) (\(i = \left\{1,\dots, t\right\}\)), 
\[C_i \coloneqq \left\{x \in \Delta \ \vert \ \mathbb{E}_x u\left(a_i,\theta\right) = V(x)\right\} \text{.}\]
Action \(a_i\) is optimal for any belief \(x \in C_i\) and uniquely optimal for any belief in the relative interior of \(C_i\) (for any \(x \in \inter C_i\)). The collection \(C\) is a regular polyhedral subdivision (henceforth, just ``subdivision'') of \(\Delta\). We say that function \(V\) induces subdivision \(C\).

\begin{definition}
    A collection of ordered expectations identifies subdivision \(C\) if there exists a \(V\) that satisfies the collection that induces \(C\), and any \(V\) that satisfies the collection of ordered expectations induces \(C\).
\end{definition}
We can always find a subdivision-identifying collection of ordered expectations.
\begin{lemma}\label{cidentify}
    For any \(C\), there exists a collection of ordered expectations that identifies it.
\end{lemma}
\begin{proof}
    Fix an arbitrary \(C\). Observe that for any element \(C_i \in C\), there exists a random vector \(X_i\) with finitely supported distribution \(F_i\) such that part of the support of \(F_i\) is the extreme points of \(C_i\), i.e., \(\extreme C_i \subseteq \supp F_i\). Take another \(\hat{X}_i\), with distribution \(\hat{F}_i\), where \(\hat{F}_i\) is a mean-preserving contraction of \(F_i\) that is identical to \(F_i\) except that the portion of the measure placed on the extreme points of \(C_i\) is collapsed to its barycenter, which by construction is in \(\inter C_i\). Then, \(\mathbb{E}_{F_i} V(X_i) = \mathbb{E}_{\hat{F}_i} V(\hat{X}_i)\) if and only if \(V\) is affine on \(C_i\). Accordingly, for a \(C\) with \(t\) elements, \(t\) such equalities implies \(V\) is affine on each \(C_i \in C\). 

    Next, for any distinct \(C_i, C_j \in C\) with \(\dim \left(C_i \cap C_j\right) = n-2\) and any \(x' \in \inter\left(C_i \cap C_j\right)\), there exists an \(X_k\) with finitely supported distribution \(F_k\) such that \(x' \in \supp F_k\). Take another \(\hat{X}_k\), with distribution \(\hat{F}_k\), where \(\hat{F}_k\) is a mean-preserving spread of \(F_k\) that is identical to \(F_k\) except that the portion of the measure placed on \(x'\) is split into two points \(x_1 \in \inter C_i\) and \(x_2 \in \inter C_j\). Then, \(\mathbb{E}_{F_k} V(X_k) < \mathbb{E}_{\hat{F}_k} V(\hat{X}_k)\) if and only if \(V\) is not affine on \(C_i \cup C_j\). As there are only finitely many pairings of distinct \(C_i\)s, finitely many such strict inequalities imply that \(V\) is not affine on any such pair.
    
    As \(V\) is convex, any \(V\) that satisfies the stated equalities and strict inequalities must induce \(C\). Moreover, for any \(C\), there is a \(V\) that induces it (indeed this is the definition of a regular polyhedral subdivision). \end{proof}
    Figures \ref{fig1} and \ref{fig2} illustrate the proof of Lemma \ref{cidentify}. In Figure \ref{fig1}, revealing that the DM (with prior on the empty-interior green point) is indifferent between the distribution over posteriors with support on the black dots and the distribution over posteriors with support on the black point in the blue region and the purple \(x\) identifies that \(V\) is affine on the red region. In Figure \ref{fig2}, revealing that the DM (with prior on the empty-interior green point) strictly prefers the distribution over posteriors with support on the black and purple \(x\)s to the distribution over posteriors with support on the black dot and the purple \(x\) identifies that \(V\) is not affine on the red and blue regions.

        \begin{figure}
        \centering
        \includegraphics[width=.5\paperwidth]{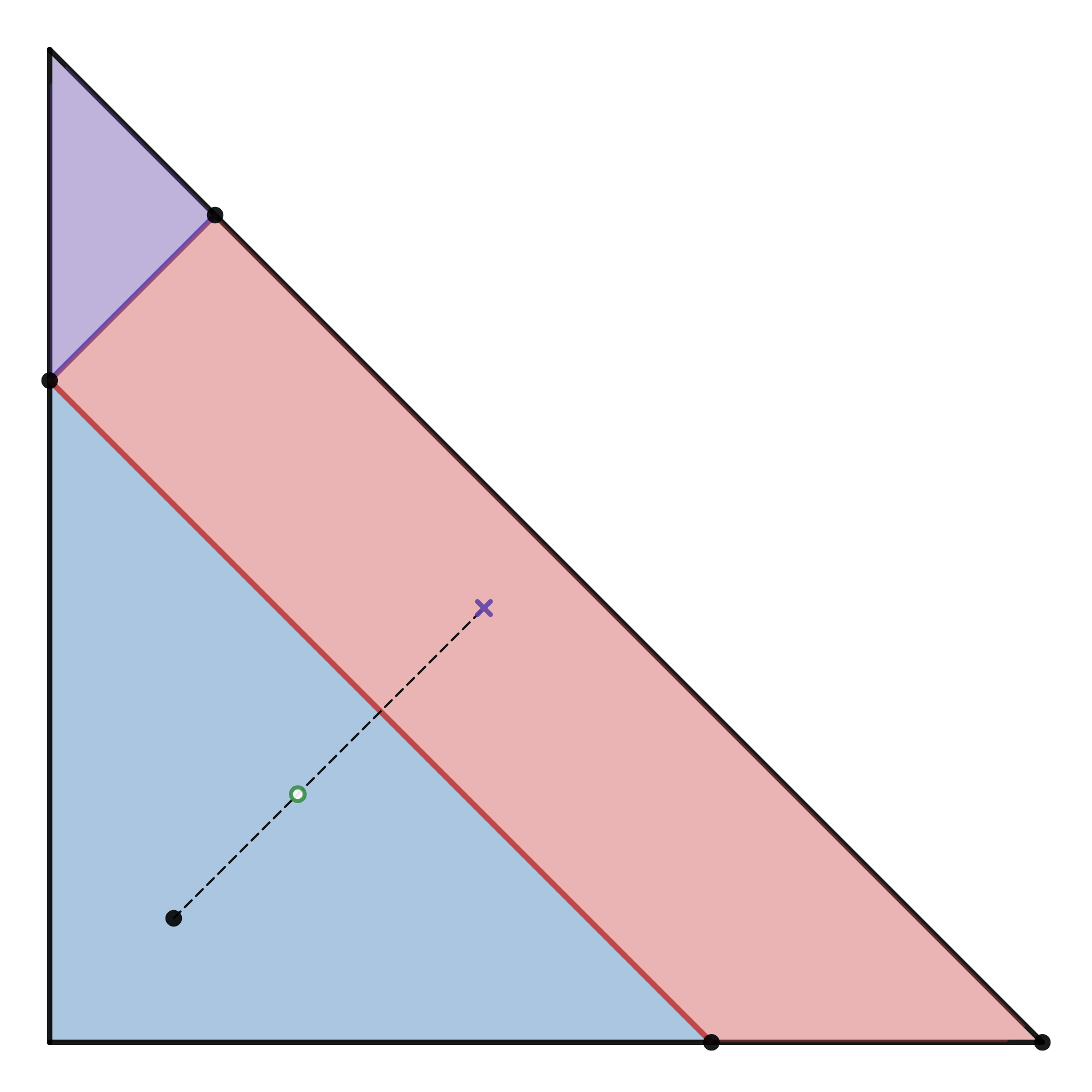}
        \caption{Showing \(V\) is affine on the red region}
        \label{fig1}
    \end{figure}

    \begin{figure}
        \centering
        \includegraphics[width=.5\paperwidth]{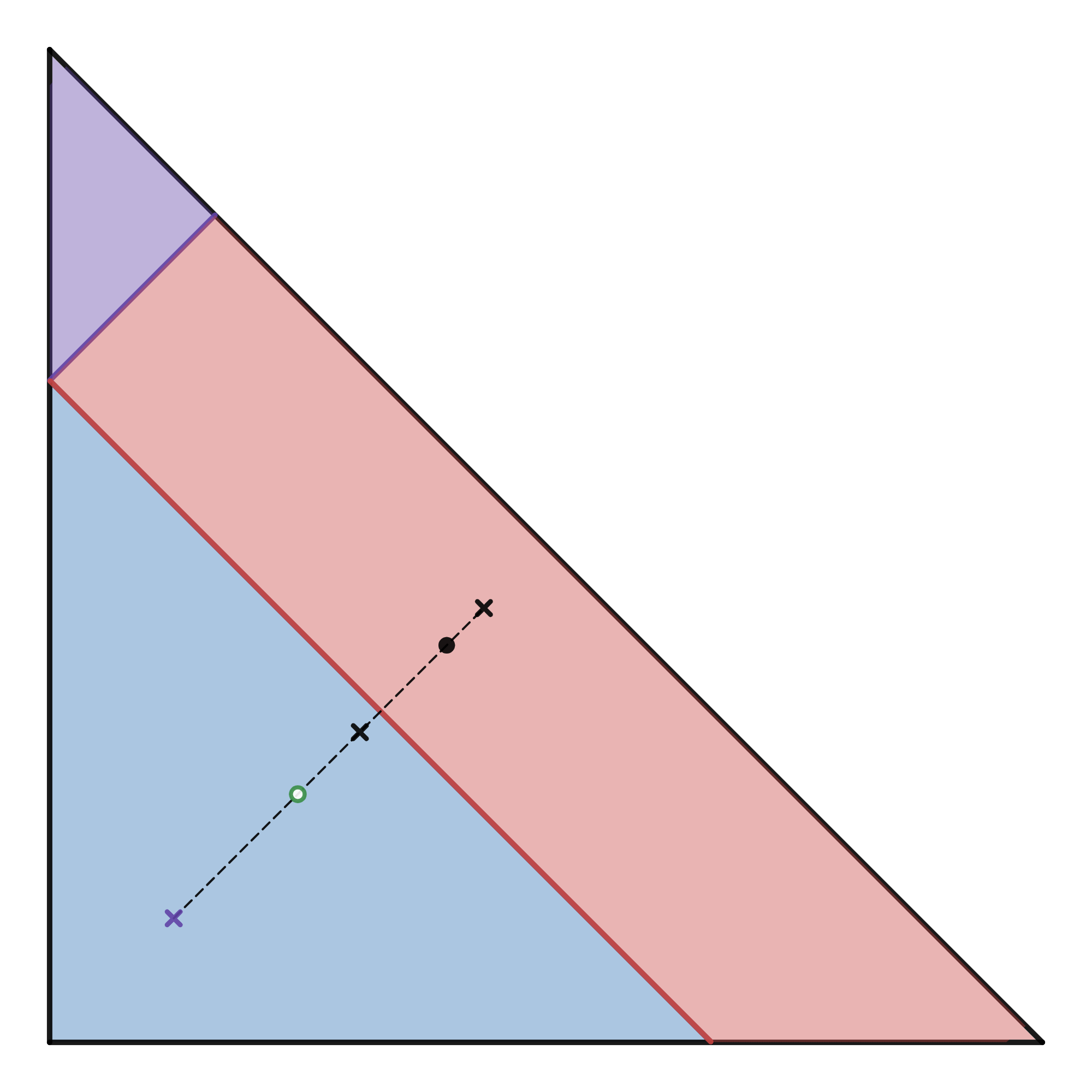}
        \caption{Showing \(V\) is not affine on the red and blue regions}
        \label{fig2}
    \end{figure}

    Next we uncover what the equivalence of value functions means in terms of the state-dependent utilities. In the next lemma, let value function \(\hat{V}\) correspond to utility \(\hat{u}\).
    \begin{lemma}
        \(\hat{V}(x) = V(x) + \varphi(x)\) for affine \(\varphi\) if and only if the sets of undominated actions for both \(\hat{u}(a,\theta)\) and \(u(a,\theta)\) are the same up to relabeling the actions; and \(\hat{u}(a,\theta) = u(a,\theta) + \gamma(\theta)\) for some \(\gamma \colon \Theta \to \mathbb{R}\) for all undominated actions up to relabeling the actions.
    \end{lemma}
    \begin{proof}
        \(\left(\Rightarrow\right)\) If \(\hat{u}(a,\theta) = u(a,\theta) + \gamma(\theta)\), then
        \[\hat{V}(x) = \max_{a \in A}\mathbb{E}_{x}\hat{u}\left(a,\theta\right) = \max_{a \in A}\mathbb{E}_{x}\left[u(a,\theta) + \gamma(\theta)\right] = V(x) + \mathbb{E}_{x}\gamma(\theta) = 
        V(x) + \varphi(x)\text{.}\]

        \smallskip

        \noindent \(\left(\Leftarrow\right)\) Suppose for the sake of contraposition that there is no \(\gamma\) such that \(\hat{u}(a,\theta) = u(a,\theta) + \gamma(\theta)\) for all undominated actions. If the subdivisions induced by \(V\) and \(\hat{V}\) are not the same, we are done (see, e.g., Lemma 2.1 in \cite*{flexibilitypaper}). Now suppose they are the same, with the same actions corresponding to each element (by relabeling). By assumption, there exist two undominated actions \(a_1\) and \(a_2\) for which
        \(\hat{u}(a_1,\theta) = u(a_1,\theta) + \gamma_{a_1}(\theta)\) and \(\hat{u}(a_2,\theta) = u(a_2,\theta) + \gamma_{a_2}(\theta)\), where \(\gamma_{a_1} \neq \gamma_{a_2}\). Thus, \[\mathbb{E}_{x}\gamma_{a_1}(\theta) = \varphi_1(x) \neq \varphi_2(x) = \mathbb{E}_{x}\gamma_{a_1}(\theta)\text{,}\]
        so
        \(\hat{V}(x) = V(x) + \varphi_1(x)\) for all \(x \in C_1\) and \(\hat{V}(x) = V(x) + \varphi_2(x)\) for all \(x \in C_2\). \end{proof}

    Now that we know we can find a collection of ordered expectations that identifies \(C\) we turn our attention to identifying \(V\). The crucial part of the proof of the next result is the observation that \(C\) pins down \(V\) (plus affine \(\varphi\)) up to one degree of freedom, but that this last remaining indeterminacy cannot be eliminated by a collection of ordered expectations.

    \begin{definition}
    A collection of ordered expectations identifies value function \(V\), if there exists a value function that satisfies the collection, and any \(W\) that satisfies the collection equals \(V\) plus some affine function \(\varphi\).
\end{definition}

\begin{proposition}\label{nonidentification}
    For any non-affine \(V\), no collection of ordered expectations identifies it.
\end{proposition}
\begin{proof}
    Fix an arbitrary \(V\). From Lemma \ref{cidentify}, there exists a collection of ordered expectations that identifies the \(C\) it induces. On any \(C_i \in C\), \(V(x) = \alpha_i \cdot x + \beta_i\). It is immediate that if \(V\) satisfies a collection of ordered expectations and induces \(C\), then so too does any \(\psi V\) (\(\psi > 0\)). As \(V\) is not affine, we are done. \end{proof}

    In some sense this result is to be expected. How could purely ordinal data identify \(V\)? Nevertheless, a few pieces of cardinal information will do. We term an equality of the form \(\mathbb{E}_{F_i} V(X_i) = \mathbb{E}_{F_j} V(X_j) + \zeta_l\) a utility difference.

    \begin{definition}
    A collection of ordered expectations and utility differences identifies value function \(V\), if there exists a value function that satisfies the collection, and any value function that satisfies the collection equals \(V\) plus some affine function \(\varphi\).
\end{definition}

    \begin{theorem}\label{theorem}
    For any \(V\), there exists a collection of ordered expectations and utility differences that identifies it.
\end{theorem}
\begin{proof}
    Fix an arbitrary \(V\) and its induced \(C\). Take any distinct \(C_i, C_j \in C\) with \(\dim \left(C_i \cap C_j\right) = n-2\). WLOG we may specify \(V = 0\) for all \(x \in C_i\), as our identification is up to the addition of an affine function. By construction, \(V(x) = \psi \left(\alpha \cdot x + \beta\right)\) (\(\alpha \in \mathbb{R}^{n-1}\), \(\beta \in \mathbb{R}\), \(\psi \in \mathbb{R}_{++}\)) for all \(x \in C_j\), with \(\alpha \cdot x + \beta \geq 0\) for all \(x \in C_j\). Moreover, \(\alpha \cdot x + \beta > 0\) for all \(x \in C_j \setminus C_i\).

    Take two points \(x_i \in \inter C_i\) and \(x_j \in \inter C_j\) such that the set \(\lambda x_i + \left(1-\lambda\right) x_j \in C_i \cup C_j\) for any \(\lambda \in \left[0,1\right]\). Let \(\hat{x} = \inter \conv \left\{x_i,x_j\right\} \cap \inter C_i\). WLOG \(\mu \in \inter \conv \left\{\hat{x},x_j\right\}\). Then, letting \(F\) be the binary distribution with support \(\left\{x_i,x_j\right\}\) and \(\hat{F}\) be the binary distribution with support \(\left\{\hat{x},x_j\right\}\), \(\mathbb{E}_F V(X) > \mathbb{E}_{\hat{F}} V(\hat{X})\). Define \(\zeta = \mathbb{E}_F V(X) - \mathbb{E}_{\hat{F}} V(\hat{X}) > 0\).

    Finally, observe that, letting \(p\) and \(q\) be defined implicitly as
    \[p x_j + (1-p) x_i = \mu = q x_j + (1-q) \hat{x}\text{,}\]
    the equation (in \(\psi\))
    \[p \psi \left(\alpha \cdot x_j + \beta\right) = q \psi \left(\alpha \cdot x_j + \beta\right) + \zeta\]
    has a unique solution 
    \[\psi = \frac{\zeta}{(p-q)\left(\alpha \cdot x_j + \beta\right)} > 0 \text{.}\] Consequently, we have identified \(V\) on \(C_j\) via this single equality. As there are only finitely many pairs of \(C_i\) and \(C_j\), we are done. \end{proof}

    As identification of \(V\) yields the value of information (in utils), \(\mathbb{E}_F V - V(\mu)\), 
    \begin{corollary}\label{corollary}
        There exists a collection of ordered expectations and utility differences that identifies the value of information for the DM.
    \end{corollary}

\section{The Irrelevance of the Prior}

Until now we have assumed that we know the DM's prior. What if we do not? Take any (full-support) prior \(\mu \in \inter \Delta\) and belief \(x \in \Delta\), that corresponds to the realization, \(s\) of a statistical experiment. It is well-known that the ratios of conditional probabilities of that \(s\) (conditioned on the states) corresponding to that belief \(x\) are well-defined functions of the prior. Consequently, for any prior \(\mu \in \inter \Delta\) and any decision problem whose value function induces \(C\), the extreme points of any \(C_i \in C\) correspond uniquely to the ratios of conditional probabilities of a number of columns (one column per extreme point of \(C_i\)) in a statistical experiment.

To put differently, for any \(\mu \in \inter \Delta\), value function \(V\) that induces \(C\), and any \(C_i \in C\), the set of extreme points of \(C_i\) corresponds uniquely to a family of sub-stochastic matrices, \(\tilde{C}_i\), defined as follows. Denoting by \(\underline{\Delta}(S)\) the set of sub-probability measures on a set of signal realizations \(S\), and \(x_s\) the DM's posterior after observing signal realization \(s\),
\[\tilde{C}_i \coloneqq \left\{\tilde{\pi} \colon \Theta \to \underline{\Delta}(S_i) \colon x_s \in \extreme (C_i) \ \forall s \in S_i, \ \cup_{s \in S_i} \left\{x_s\right\} = \extreme (C_i) \right\}\text{.}\]

By construction, for a decision problem whose value function induces subdivision \(C\), at prior \(\mu \in \inter \Delta\), the DM is indifferent between any experiment \(\pi\) possessing \(\tilde{\pi} \in \tilde{C}_i\) as a submatrix and any experiment \(\tilde{\pi}'\) that is the same as \(\pi\) except that \(\tilde{\pi}\) is replaced by an arbitrary garbling. We term \(\tilde{C}_i\) a spectral element, and the set \(\tilde{C} \coloneqq \left\{\tilde{C}_1, \dots, \tilde{C}_t\right\}\) the spectral subdivision. Importantly, a spectral subdivision is prior-free: it is a set of families of matrices.

We call a finite collection of Blackwell experiments \(\pi_1, \dots, \pi_m\) and rankings of the form \(\pi_i \succ \pi_j\) and \(\pi_k \sim \pi_l\) (\(i, j, k, l \in \left\{1, \dots, m\right\}\)) a collection of ranked experiments. 

\begin{definition}
    A collection of ranked experiments identifies spectral subdivision \(\tilde{C}\) if for any prior \(\mu \in \inter \Delta\) there exists a decision problem that satisfies the collection that induces \(\tilde{C}\), and any decision problem and prior that satisfies the collection induces \(\tilde{C}\).
\end{definition}

We immediately deduce a corollary of Lemma \ref{cidentify}.
\begin{corollary}\label{spectral}
    For any spectral subdivision \(\tilde{C}\), there exists a collection of ranked experiments that identifies it.
\end{corollary}

We can also add some cardinal information.  Recall that for a Blackwell experiment \(\pi\), \(W(\pi) \in \mathbb{R}\) denotes the DM's Expected payoff in her decision problem from first observing the outcome of experiment \(\pi\) then taking an optimal action. We term an equality of the form \(W(\pi_i) = W(\pi_j) + \zeta_l\) a utility difference. Recall from Corollary \ref{corollary} that if we knew \(\mu\), we could identify the DM's value for information. It turns out, we do not even need to know \(\mu\) to price (in utils) information for the agent. That is, the precise utility gain from \textit{any} experiment is identified. 

We normalize the value of no information to \(0\) so that the value of information of experiment \(\pi\) is merely \(W\left(\pi\right)\). Then,

\begin{definition}
    A collection of ranked experiments and utility differences identifies the value of information \(W\) if for any prior \(\mu \in \inter \Delta\) there exists a decision problem such that the value of an experiment \(\pi\) is \(W(\pi)\), and any decision problem and prior satisfying the collection has value of information \(W\).
\end{definition}

\begin{theorem}
    There exists a collection of ranked experiments and utility differences that identifies the value of information for the agent.
\end{theorem}
\begin{proof}
    Take an arbitrary prior \(\mu' \in \inter \Delta\) and a decision problem that induces value function \(V\). By Theorem \ref{theorem} there exists a collection of ordered expectations and utility differences that identifies it and, hence, \(W\) for this DM. Moreover, take the collections specified in the proofs of Proposition \ref{cidentify} and Theorem \ref{theorem}, and observe that the ranked experiments corresponding to the ordered expectations and the utility differences identify value function \(V^\mu\) for each \(\mu \in \inter \Delta\). Moreover, by construction, for all such priors \(\mu\), the value of information in the corresponding decision problem (with value function \(V^\mu\)) is the same and, thus, equal to \(W\).

    Accordingly, it suffices to show that for an arbitrary prior \(\mu' \in \inter \Delta\), there exists a collection of ranked experiments and utility differences that identifies value of information for the agent. This is Corollary \ref{corollary}, so we are done. \end{proof}

\section{Identifying a Two-Player Game}

There is a tight connection between two-player (simultaneous move) games and decision problems. In such games, the other player takes on the role of nature, generating any uncertainty by randomizing over his strategies. Given this, it is easy to construct analogs of the earlier identification results.

Formally, a finite two-player (normal form) game is \(G = \left(\left\{1,2\right\}, A, u\right)\) consisting of two players (\(1\) and \(2\)), a finite set of (pure) action profiles \(A = A_1 \times A_2\) and a profile of utilities for the players \(u = \left(u_1,u_2\right)\), with \(u_i \colon A \to \mathbb{R}\) (\(i = 1, 2\)). \(\Sigma = \Sigma_1 \times \Sigma_2\) denotes the set of (mixed) strategy profiles, with \(\Sigma_i = \Delta A_i\) (\(i = 1,2\)). As above, we stipulate that there are no duplicate actions for either player. As before, this is an imposition of genericity.

We define player \(1\)'s best response correspondence, in \(\sigma_2 \in \Sigma_2\), \(b_1\left(\sigma_2\right)\), in the standard manner:
\[b_1\left(\sigma_2\right) \coloneqq \left\{a_1 \in A_1 \colon \mathbb{E}_{\sigma_2}u\left(a_1,a_2\right) = \max_{a_1' \in A_1} \mathbb{E}_{\sigma_2}u\left(a_1',a_2\right)\right\}\text{.}\]
We define player \(2\)'s best response correspondence \(b_2\left(\sigma_1\right)\) analogously. We term a finite collection of the form \(b_1\left(\sigma_2^i\right) \cap b_1(\sigma_2^j) \neq \emptyset\), \(b_2\left(\sigma_1^i\right) \cap b_2(\sigma_1^j) \neq \emptyset\), \(b_1\left(\sigma_2^k\right) \cap b_1\left(\sigma_2^l\right) = \emptyset\), and \(b_2\left(\sigma_1^k\right) \cap b_2\left(\sigma_1^l\right) = \emptyset\) a collection of best-response comparisons.

As in the decision problem above, \(\Sigma_2\) is subdivided into a finite collection \(C_1\) of polytopes of full dimension \(C_1^i\), where for each undominated \(a_1 \in A_1\), 
\[C_1^i \coloneqq \left\{\sigma_2 \in \left.\Sigma_2 \right| a_1^i \in b_1(\sigma_2)\right\}\text{.}\]
\(\Sigma_1\) is likewise subdivided into collection \(C_2\). We call \(C_i\) player \(i\)'s subdivision and say that the game induces subdivision pair \((C_1, C_2)\) if the game is such that player \(1\)'s subdivision is \(C_1\) and player \(2\)'s subdivision is \(C_2\).
\begin{definition}
    A collection of best-response comparisons identifies game \(G\) if there exists a game whose best-response correspondences satisfy the collection, and any game whose best-response correspondences satisfy the collection induces the same subdivision pair \((C_1, C_2)\).
\end{definition}
Mirroring the argument for Lemma \ref{cidentify}, we have
\begin{theorem}
    For any game there exists a collection of best-response comparisons that identifies it.
\end{theorem}
\begin{proof}
    As the proof is virtually identical to that for Lemma \ref{cidentify}, we merely sketch it.\footnote{In fact it is even easier, as we are no longer constrained by Bayes' law.} Take an arbitrary element \(C_1^i \in C_1\). It is a polytope with a finite set of extreme points (its vertices) \(\left\{\sigma_2^1, \sigma_2^2, \dots, \sigma_2^k\right\}\). Accordingly, for some \(\sigma_2^* \in \inter C_1^i\) the collection
    \[b(\sigma_2^1) \cap b(\sigma_2^2) \cap \dots \cap b(\sigma_2^k) \cap b(\sigma_2^*) \neq \emptyset\] implies that there is an action that is a best response to any \(\sigma_2 \in C_1^i\). There are only finitely many elements in \(C_1\), so finitely many more such collections imply that there are actions that are best responses, respectively, to any \(\sigma_2 \in C_1^i\) for any \(C_1^i \in C_1\). Finally, for any distinct pair \(C_1^i\) and \(C_1^j\) (of which there are finitely many) a comparison 
    \[b(\sigma_2^\heartsuit) \cap b(\sigma_2^\diamond) = \emptyset \text{,}\]
    where \(\sigma_2^\heartsuit \in \inter C_1^i\) and  \(\sigma_2^\diamond \in \inter C_1^j\), distinguishes between the respective best responses.\end{proof}

\section{Related Work}

This paper belongs among those studying the value of information in decision problems. That literature includes \cite*{ram}, \cite*{athey}, and \cite*{azrieli}; as well as \cite*{radstig}, \cite*{delara}, \cite*{Chade}, and \cite*{dentist}, which four all study the marginal valuation of information. We also make use of the connection between decision problems and polyhedral subdivisions. \cite*{greenosband} reveal this connection, and \cite*{lambert} and \cite*{frongillo2021general} make use of it when studying elicitability of various forms of information.

A number of other works seek to understand what information a DM obtains based on her choice behavior in a (known) decision problem. \cite*{dillenberger2014theory} and \cite*{lu2019bayesian} are two such papers. \cite*{de2017rationally} and \cite*{caplin2015revealed} are two others, both of which endogenize the DM's information structure. \cite*{libgober2021hypothetical} shows how to identify information structures from a population's beliefs.

Finally, this work is fundamentally a revealed preference exercise (\cite*{samuelson1938note} and \cite*{afriat1967construction}): given a choice paradigm what can a set of observed choices by an agent tell you about the agent (and the model)? More recently, \cite*{caplinmartin} and \cite*{caplin2023rationalizable} ask what choice data reveals about the information a DM could have obtained and \cite*{caplin2015revealed} and \cite*{chambers2020costly} whether a DM's behavior is consistent with costly information acquisition. This paper, instead, asks what preferences over information structures reveal about a DM's choice problem.

\bibliography{sample.bib}

\end{document}